\documentclass[11pt,twoside]{article}
\usepackage{amssymb}
\usepackage{mathtools,extarrows}
\usepackage{amsmath,bm,amsfonts,color,amsthm, amssymb}
\usepackage[margin=3cm]{geometry}
\usepackage{graphicx}
\usepackage[sort&compress]{natbib}
\usepackage{comment}
\usepackage{scrextend}
\usepackage{lipsum}
\bibpunct{(}{)}{;}{a}{}{,}
\usepackage[colorlinks,citecolor=blue,urlcolor=black]{hyperref}
\usepackage{lineno}
\usepackage{eurosym}
\newcommand\numberthis{\addtocounter{equation}{1}\tag{\theequation}}
\usepackage{bbm, dsfont}
\setcitestyle{numbers}
\usepackage{lipsum}
\usepackage{fancyhdr}

\usepackage{tikz}
\usetikzlibrary{decorations.pathmorphing} 
\usetikzlibrary{fit}         
\usetikzlibrary{backgrounds} 
\usetikzlibrary{positioning}
\usetikzlibrary{shadows}
\usepgflibrary{shapes}
\tikzstyle{observed}=[circle, thick, minimum size=0.7cm, draw=black!100, fill=black!20]
\tikzstyle{latent}=[circle, thick, minimum size=0.7cm, draw=black!100]
\tikzstyle{plate}=[rectangle, thick, inner sep=0.3cm, draw=black!100]
\tikzstyle{shadeplate}=[rectangle, thick, inner sep=0.3cm, draw=black!100, fill=black!10]
\numberwithin{equation}{section}
\theoremstyle{plain}

\newtheorem{lemma}{Lemma}[section]

\title{\textbf{Empirical and Constrained Empirical Bayes Variance Estimation Under A One Unit Per Stratum Sample Design}}

\author{Sepideh Mosaferi \textsuperscript{1}}
	\date{
	\textsuperscript{1} University of Maryland College Park; 
	\url{{smosafer}@umd.edu} \\
	\ December 15, 2014}

\newcommand\shorttitle{Empirical and Constrained EB Variance Estimators}
\newcommand\authors{Mosaferi, S.}
	
\begin{document}

\maketitle

\begin{abstract}
\noindent A single ``primary sampling unit (PSU)" per stratum design is a popular design for estimating the parameter of interest. Although, the point estimator of the design is unbiased and efficient, an unbiased variance estimator does not exist. A common practice to solve this is to \textit{collapse} or \textit{combine} the two adjacent strata, but the attained estimator of variance is not design-unbiased, and the bias increases as the population means of the collapsed strata become more variant. Therefore, the \textit{one} PSU per stratum design with collapsed stratum variance estimator might not be a good choice, and some statisticians prefer a design in which \textit{two} PSUs per stratum are selected. In this paper, we first compare a one PSU per stratum design to a two PSUs per stratum design. Then, we propose an empirical Bayes estimator for the variance of one PSU per stratum design, where it over-shrinks towards the prior mean. To protect against this, we investigate the potential of a constrained empirical Bayes estimator. Through a simulation study, we show that the empirical Bayes and constrained empirical Bayes estimators outperform the classical collapsed one in terms of empirical relative mean squared error.
\vspace{0.25cm}


\end{abstract}
\thispagestyle{empty}

\section{Introduction} \label{intro}

A design in which one primary sampling unit (PSU) is selected in each stratum is theoretically efficient for providing an unbiased estimator of a population parameter. However, estimation of the variability of the attained estimator is impossible without considering any implicit assumptions such as collapsing strata; such assumptions produce a \textit{design-biased} estimator of the variance. Some examples that use stratified multi-stage with one PSU per stratum design include 
the Current Population Survey (CPS) and the National Crime Victimization Survey in the United States.

Due to the lack of an unbiased variance estimator for the one PSU per stratum design, some survey statisticians prefer to select two PSUs per stratum since the variance estimators for simple estimators such as the Horvitz-Thompson estimator are unbiased.
Surveys such as the Survey of Income and Program Participation (SIPP) of the U.S. Census Bureau and the U.S. Department of Agriculture's National Resources Inventory use a multi-stage two-per-stratum design and a stratified two-stage area sampling design, respectively. For the design with two PSUs per stratum, an unbiased variance estimator for the linear estimators exist without any implicit assumptions. Although, the one PSU per stratum design still has its own popularity as it allows \textit{deep stratification}.

The collapsed stratum method for variance estimation was first introduced by [\citep*{hansen1953sample}]. This method usually causes an overestimation in the variance of estimator; therefore, [\citep*{hansen1953sample}] and [\citep*{isaki1983variance}] proposed to use some auxiliary variables well-correlated with the expected values of the mean of stratums to reduce the bias of variance estimator.
[\citep*{hartley1969variance}] proposed a method of grouping strata where each group contains 7 to 15 strata and then applied a linear regression of the group means on one or more auxiliary variables for estimating the variance of one PSU per stratum design where regression residuals used to estimate the components of variance. This method requires a further evaluation study before being used, and the bias of the variance estimator depends on how well the regression model fits. 

The idea of stratum boundaries being chosen by a random process prior to the sample selection was proposed by [\citep*{fuller1970sampling}]. Method [\citep*{fuller1970sampling}] is biased when the stratum boundaries are not randomized beforehand.
[\citep*{rust1987strategies}] examined the effects of collapsing strata in pairs, triples, and larger groups on the quality of the variance estimator and found that a greater level of collapsing is desirable when a small sample of PSUs is selected. [\citep*{rust1987strategies}] provided a list of factors which might help to decide on the extent of collapsing. 

Under the assumption of [\citep*{durbin1967design}] for the sampling scheme within the collapsed strata, [\citep*{shapiro1978better}] applied the method of [\citep*{yates1953selection}] for the variance estimator. This variance can be biased upward, but the bias is relatively smaller than the collapsed method, and the variance estimator is more stable. However, for their empirical example, the authors did not consider the same number of units in the collapsed strata. This results the collapsed stratum variance performs poorly compared to the situation of having the same number of units in each collapsed strata.

[\citep*{mantel2009variance}] proposed a new approach based on the components of variance from different stages of sampling. They studied the Canadian Health Measures Survey (CHMS), a three-stage sample design where the number of PSUs per stratum is very small. In the study, they assumed that a randomized PPS systematic (RPPSS) sampling design used for the CHMS. This assumption might introduce an unknown bias into the variance estimation. They also could not calculate an uncollapsed variance estimate for the Atlantic, the stratum with one PSU.

Recently, [\citep*{breidt2016nonparametric}] proposed a nonparametric alternative method that replaces a collapsed stratum estimator by kernel-weighted stratum neighborhoods and used deviations from a fitted mean function to estimate the variance. They applied their method to the U.S. Consumer Expenditure Survey to demonstrate the superiority of their method over the collapsed stratum variance estimator. The estimator that they used is a natural nonparametric extension of linear models proposed by [\citep*{hartley1969variance}] and [\citep*{isaki1983variance}].

In fact, most of the recommended alternative methods for the collapsed stratum variance are based on the existence of some concomitant or auxiliary information; nevertheless, this kind of desirable auxiliary information might not be readily available for all of the strata. So, finding an acceptable comparative variance estimator for the one PSU per stratum design is still under question.

The rest of manuscript is organized as follows. In section \ref{comparison}, we systematically compare the one PSU per stratum design [Design 1] to the two PSUs per stratum desgin [Design 2] based on the actual variance of the point estimator of population mean through a simulation study. In section \ref{variance}, we analytically compare the two design strategies as well as through a Monte Carlo simulation study with respect to the coverage probability of mean.
In section \ref{EBandCEB}, we propose two alternatives to the collapsed variance estimator for the one PSU per stratum design using an empirical Bayes and a constrained empirical Bayes approaches.
The results are summarized in section \ref{allsimulations}, followed by main remarks in section \ref{remarks}, and we defer all of the necessary proofs to the section \ref{appendix}.

\section{Comparison of Design 1 and Design 2} \label{comparison}

For simplicity of exposition, we consider a stratified design with $H$ strata where stratum $h (h=1,...,H)$ consists of $N_h$ units. A sample of $n_h$ units is selected from each stratum $h$ through a simple random sampling without replacement. In many applications, these units could be the primary stage units. In this paper, we concentrate on single-stage sampling, and we are interested in estimating the finite population mean $\bar{Y}=\sum_{h=1}^{H}W_h \bar{Y}_h$, where $W_h=N_h/N_T$, and $N_T$ is $\sum_{h=1}^{H}N_h$. Note that $\bar{Y}_h$ is the finite population mean of the $h$-th stratum. The Horvitz-Thompson unbiased estimator of the finite population mean ($\bar{Y}$) is given by
$ \bar{y}_{st} =\sum_{h=1}^{H}W_h\bar{y}_h$,  
where $\bar{y}_h$ is the sample mean for the $h$-th stratum. 
The associated randomization-based variance is given by:

\begin{equation}
\label{eq:2.1}
V(\bar{y}_{st})=\sum_{h=1}^{H}W_h^2\frac{1}{n_h}(1-\frac{n_h}{N_h})S_h^2, 
\end{equation}
where $S_h^2=\sum_{j=1}^{N_h}(y_{hj}-\bar{Y}_h)^2/(N_h-1)$ is the finite population variance for the $h$-th stratum.

Here, we compare two popular design options: 
\begin{itemize}
\item[] $n_h=1$: one PSU per stratum design, called [Design 1], and
\item[] $n_h=2$: two PSUs per stratum design, called [Design 2].
\end{itemize}
These two options are widely used in the context of stratified cluster sampling and stratified multi-stage sampling designs. To make a fair comparison of Design 1 and Design 2, we consider $2H$ strata for Design 1 and $H$
groups, each with two strata, for Design 2; therefore, we have an equal number of units (PSUs) for both designs.

The relative efficiency of Design 1 relative to Design 2 can be measured by the design effect deff$=V_2(\bar{y}_{st})/V_1(\bar{y}_{st})$, where $V_1$ and $V_2$ are the randomization-based variances of the same point estimator $\bar{y}_{st}$ mentioned in \ref{eq:2.1}. If deff$>$1, Design 1 is more efficient than Design 2. On the other hand, if deff$<$1, Design 1 is less efficient than Design 2. If deff=1, the two designs are equivalent.

To compare Design 1 with Design 2, we conduct a Monte Carlo simulation study. Following [\citep*{hansen1983evaluation}], we generate a finite population of size $N_T=20,000$ units by drawing a random sample of size 20,000 from a bivariate superpopulation characterized by the two dimensional random vector $(x,y)$, where the variable $x$ has a gamma distribution with shape 2 and scale 5, i.e. $f(x)=.04x \exp(-x/5)$, and the variable $y$ conditional on $x$ (i.e. $y|x$) has a gamma distribution with density function $g(y;x)=(\frac{1}{b^c} \Gamma (c)) y^{c-1} \exp(-y/b)$, where $c=.04 x^{-3/2} (8+5x)^2$ and $b=1.25 x^{3/2}(8+5x)^{-1}$. 

To compare the effects of the number of strata $H$ on the relative efficiency, we consider $H=10, 50,$ and $100$ strata, which are formed based on the quantiles of $x$. We display the results in Table \ref{table1}. Overall, Design 1 performs better than Design 2 with considerable efficiency for small $H$. However, efficiency diminishes as the number of strata increases (see Table \ref{table1}).

\begin{table}[ht]
\caption{Comparison of Design 1 and Design 2 Based on the Number of Strata} \label{table1}
\centering
\begin{tabular}{@{}ccccccccc@{}}
\hline\hline
 \multicolumn{3}{c}{Design 1} & & \multicolumn{3}{c}{Design 2} & & Comparison\\
\cline{1-3} \cline{5-7} 
 $H$ & $N_h$ & $V_1(\bar{y}_{st})$ & & $H$ & $N_h$ & $V_2(\bar{y}_{st})$ & & deff=$V_2(\bar{y}_{st})/V_1(\bar{y}_{st})$\\\hline
10 & 2000 & 0.2515 &  & 5 & 4000 & 0.2759 & & 1.0969 \\
50 & 400 & 0.0464 & & 25 & 800 & 0.0469  & & 1.0104 \\
100 & 200 & 0.0229  & & 50 & 400 & 0.0232 & & 1.0109 \\
\hline
\end{tabular}
\end{table}

To be able to assess the effects of differences in population means and population variances within the collapsed strata on the randomization-based variance \ref{eq:2.1}, 
we considered $H=10$ strata and $5$ strata for Designs 1 and 2, respectively, and employed some changes to the generated population's means and variances within and between groups (collapsed strata).
We separately generated data for each stratum from the Normal distribution, $N(\mu=mean(y),\sigma^2=var(y))$ and considered different coefficients of k, 2k, 3k, 4k, 5k (k = 1, 2) to implant some changes in $\mu$ and $\sigma^2$ for groups $(g)$ 1, 2, 3, 4, and 5, respectively. 
The populations used for generating the data based on the Normal distribution are given in Table \ref{table2}, and the results of comparisons associated to the groups of Table \ref{table2} are in Table \ref{table3}.
\setlength{\tabcolsep}{3pt}

\begin{table}[ht]
\caption{Generated Populations from the Normal Distribution for the Comparison} \label{table2}
\small
\centering
\begin{tabular}{cccccc}
\hline\hline
Case Study & Group 1 & Group 2 & Group 3 & Group 4 & Group 5 \\\hline
\\
$\bar{Y}_{g1}\approx\bar{Y}_{g2}$ & $N(\mu+1,\sigma^2)$ & $N(\mu+2,2\sigma^2)$ & $N(\mu+3,3\sigma^2)$
& $N(\mu+4,4\sigma^2)$ & $N(\mu+5,5\sigma^2)$\\
$s^2_{g1}\approx s^2_{g2}$ & $N(\mu+1,\sigma^2)$ & $N(\mu+2,2\sigma^2)$ & $N(\mu+3,3\sigma^2)$ &
$N(\mu+4,4\sigma^2)$ & $N(\mu+5,5\sigma^2)$ \\
&&&&&\\
$\bar{Y}_{g1}\neq\bar{Y}_{g2}$ & $N(\mu+1,\sigma^2)$ & $N(\mu+2,2\sigma^2)$ & $N(\mu+3,3\sigma^2)$
& $N(\mu+4,4\sigma^2)$ & $N(\mu+5,5\sigma^2)$\\
$s^2_{g1}\approx s^2_{g2}$ & $N(\mu+2,\sigma^2)$ & $N(\mu+4,2\sigma^2)$ & $N(\mu+6,3\sigma^2)$ &
$N(\mu+8,4\sigma^2)$ & $N(\mu+10,5\sigma^2)$ \\
&&&&& \\
$\bar{Y}_{g1}\neq\bar{Y}_{g2}$ & $N(\mu+1,\sigma^2)$ & $N(\mu+2,2\sigma^2)$ & $N(\mu+3,3\sigma^2)$
& $N(\mu+4,4\sigma^2)$ & $N(\mu+5,5\sigma^2)$\\
$s^2_{g1}\neq s^2_{g2}$ & $N(\mu+2,2\sigma^2)$ & $N(\mu+4,4\sigma^2)$ & $N(\mu+6,6\sigma^2)$ &
$N(\mu+8,8\sigma^2)$ & $N(\mu+10,10\sigma^2)$ \\
&&&&& \\
$\bar{Y}_{g1}\approx\bar{Y}_{g2}$ & $N(\mu+1,\sigma^2)$ & $N(\mu+2,2\sigma^2)$ & $N(\mu+3,3\sigma^2)$
& $N(\mu+4,4\sigma^2)$ & $N(\mu+5,5\sigma^2)$\\
$s^2_{g1}\neq s^2_{g2}$ & $N(\mu+1,2\sigma^2)$ & $N(\mu+2,4\sigma^2)$ & $N(\mu+3,6\sigma^2)$ &
$N(\mu+4,8\sigma^2)$ & $N(\mu+5,10\sigma^2)$ \\ \\
\hline
\end{tabular}
\end{table}

According to Table \ref{table3}, when population means within groups are different, Design 1 is more efficient than Design 2; on the other hand, when the population means within groups are similar, there is no preference between the two designs. 
In addition, changes in the population variances within the groups do not show any conspicuous effects on the efficiency.

\begin{table}[ht]
\caption{Comparison of Design 1 and Design 2 Based on the Differences of Means and Variances} \label{table3}
\centering
\begin{tabular}{ccccccccccc}
\hline\hline
  & & \multicolumn{3}{c}{Design 1} & & \multicolumn{3}{c}{Design 2} & & Comparison\\
\cline{3-5} \cline{7-9} 
Case Study & & $H$ & $N_h$ & $V_1(\bar{y}_{st})$ & & $H$ & $N_h$ & $V_2(\bar{y}_{st})$ & & deff=$V_2(\bar{y}_{st})/V_1(\bar{y}_{st})$\\\hline \\

$\bar{Y}_{g1}\approx\bar{Y}_{g2}$ and $S_{g1}^2 \approx S_{g2}^2$ &&10 & 2000 & 1.6198  & & 5 & 4000 & 1.6198  & & 1.0000 \\ \\

$\bar{Y}_{g1}\neq\bar{Y}_{g2}$ and $S_{g1}^2 \approx S_{g2}^2$ &&10 & 2000 & 1.6197  & & 5 & 4000 & 1.9021 & & 1.1743 \\ \\

$\bar{Y}_{g1}\neq\bar{Y}_{g2}$ and $S_{g1}^2 \neq S_{g2}^2$ &&10 & 2000 & 2.4310  & & 5 & 4000 & 2.7047 & & 1.1126 \\ \\

$\bar{Y}_{g1} \approx\bar{Y}_{g2}$ and $S_{g1}^2 \neq S_{g2}^2$ &&10 & 2000 & 2.4476  & & 5 & 4000 & 2.4476 & & 1.0000 \\ \\
\hline
\end{tabular}
\end{table}

\section{Variance Estimation in Design 1 and Design 2} \label{variance}

\subsection{Theoretical Expressions} \label{theory}

As in section \ref{comparison}, we assume that we have two strata in each of the $h$-th group and let $N_{gi}$ denote the population size for the $i$-th stratum within the $g$-th group. Let $y_{gij}$ denote the value of the characteristic of interest for the $j$ unit in the $i$ stratum within the $g$ group ($g=1,\cdots,H,\;i=1,2,\;j=1,\cdots,N_{gi}$). For simplicity in exposition, we assume $N_{gi}=N$ ($\forall g=1,\cdots,H,\;i=1,2$), therefore $N_T=2HN$ and $W_{gi}=N_{gi}/N_T=1/2H$.  
We further define:
\begin{enumerate}
\item[] $\bar Y_{gi}=N^{-1}\sum_{j=1}^Ny_{gij}:$ finite population mean for the $i$-th stratum within the $g$-th group,
\item[] $S_{gi}^2=(N-1)^{-1}\sum_{j=1}^N(y_{gij}-\bar Y_{gi})^2:$ finite population variance for the $i$-th stratum in the $g$-th group,
\item[] $\mu_{r,gi}=(N-1)^{-1}\sum_{j=1}^N(y_{gij}-\bar Y_{gi})^r:$ finite population $r$-th central moment ($r\ge 1$). Note that $\mu_{1,gi}=0$ and $\mu_{2,gi}=S_{gi}^2.$
\end{enumerate}
We also assume the finite population correction (FPC) factor is negligible as the sample size is only one or two per stratum, and $N$ is large.

The true variance based upon Design 1 is:
\begin{equation} \label{eq:3.1}
V(\bar{y}_{st})=\frac{1}{4H^2} \sum_{h=1}^{2H}S_h^2.
\end{equation}
where $S^2_h=\sum_{j=1}^{N}(y_{hj}-\bar{Y}_h)^2/(N-1)$. We rewrite \ref{eq:3.1} as
$
V(\bar{y}_{st})=\frac{1}{4H^2}\sum_{g=1}^{H}(S^2_{g1}+S^2_{g2}).
$
and the collapsed strata variance estimator is given by:
 \begin{equation} \label{eq:3.2}
v(\bar{y}_{st})=\frac{1}{2H^2}\sum_{g=1}^{H}s_g^2, 
 \end{equation}
 where $s^2_g=\sum_{i=1}^{2}(y_{gi}-\bar{y}_g)^2$, and $\bar{y}_g=(y_{g1}+y_{g2})/2$.
 The method relies on the implicit assumption of $\bar Y_{g1}=\bar Y_{g2}=\bar Y_{g}$.
 
Estimator \ref{eq:3.2} is design-biased, and its bias with respect to Design 1 is given by:
\begin{equation} \label{eq:3.3}
Bias(v(\bar{y}_{st}))=\frac{1}{4H^2}\sum_{g=1}^{H}(\bar{Y}_{g1}-\bar{Y}_{g2})^2.
\end{equation}
As it is clear from \ref{eq:3.3}, the bias is not related to the population variances within the groups.
 [\citep*{wolter1985introduction}] computed the bias of population total given the original sampling design. 
Equation \ref{eq:3.3} suggests the strategy of how we can group strata to reduce the bias of collapsed stratum variance by putting more similar strata in pairs with respect to the characteristic of interest to minimize the difference 
$|\bar{Y}_{g1}-\bar{Y}_{g2}|$.

In order to find out the mean squared error (MSE) of $v(\bar{y}_{st})$, the theoretical variance of variance is needed. By ignoring the FPC, the variance is:

\begin{align*} \label{eq:3.4}
Var(v(\bar{y}_{st})) & =\frac{1}{16H^4}\sum_{g=1}^{H}\{\mu_{4,g1}
+\mu_{4,g2}+2S^2_{g1}S^2_{g2}+4(\bar{Y}_{g1}-\bar{Y}_{g2})^2 (S^2_{g1}+S^2_{g2}) \\
& \quad -(S^2_{g1}-S^2_{g2})^2
+4(\bar{Y}_{g1}-\bar{Y}_{g2})(\mu_{3,g1}-\mu_{3,g2})\}. \numberthis
\end{align*}
If $\mu_{4,g1}=\mu_{4,g2}=\mu_{4,g}$, $\mu_{3,g1}=\mu_{3,g2}=\mu_{3,g}$, and $S^2_{g1}=S^2_{g2}=S^2_g$, then
\begin{equation*}
Var(v(\bar{y}_{st}))=\frac{1}{8H^4}\sum_{g=1}^{H}\{
\mu_{4,g}+(S^2_g)^2+4S^2_g(\bar{Y}_{g1}-\bar{Y}_{g2})^2\}. 
\end{equation*}
Therefore, the \textit{MSE} of $v(\bar{y}_{st})$ under all of these equality assumptions is:
\begin{align*}
 MSE(v(\bar{y}_{st})) & =Var(v(\bar{y}_{st}))+\{Bias(v(\bar{y}_{st}))\}^2 \\
 &  =\frac{1}{8H^4}\sum_{g=1}^{H}\{
\mu_{4,g}+(S^2_g)^2+4S^2_g(\bar{Y}_{g1}-\bar{Y}_{g2})^2\}
 +\frac{1}{16H^4}\sum_{g=1}^{H}(\bar{Y}_{g1}-\bar{Y}_{g2})^4. 
\end{align*}

Observe the MSE is inversely related to the number of strata $H$. As a result,  according to the asymptotic properties we can expect as the number of strata $H$ increases, MSE decreases. 
For Design 2, we ignore the FPC and use the standard variance of stratified estimator, which is unbiased under the design (see, [\citep*{hansen1953sample}] as an example).

\subsection{Simulation Study} \label{simulation}

We performed a simulation experiment to investigate the differences between the two designs with respect to the empirical coverage probability (CP) and average length (AL) of a nominal $95\%$ confidence interval (CI) for $\bar{y}_{st}$ under the two designs. 
The population used for this sub-section is similar to the one used in Table \ref{table3}.
The sample designs are the random selection of 1 PSU and two PSUs without replacement in each stratum for Design 1 and Design 2, respectively.

 The process of sample selection was repeated 10,000 times, and for each replication, we obtain $\bar{y}_{st}$ and the two-sided $95\%$ confidence interval, $\bar{y}_{st} \pm 1.96 \sqrt{v(\bar{y}_{st})}$.
For the variance, we use the standard unbiased variance estimate for the two PSUs per statum design and the collapsed strata variance given by \ref{eq:3.2} for the one PSU per stratum design. Table \ref{table4} displays the empirical CP and average CI for the both designs.

From the given results, we observe when the means of collapsed strata are similar, the two designs perform almost identical (see Table \ref{table4}). On the other hand, the AL and CP are greater when the population means of collapsed strata becomes more different; this can reflect the important effect of $|\bar{Y}_{g1}-\bar{Y}_{g2}|$
in collapsing. Based on Table \ref{table4}, we cannot say which design is better when $|\bar{Y}_{g1}-\bar{Y}_{g2}|\neq0$, since the more variability in Design 1 might result into the greater CP. In addition, as the number of strata increases the CP for the both designs approaches the nominal coverage probability 0.95; however, the results are not displayed in Table \ref{table4} to save space and to be consistent with other tables.

\begin{table}[ht]
\caption{Empirical Results of Simulation Study for Comparison of Design 1 and Design 2} \label{table4}
\centering
\begin{tabular}{ccccc}
\hline\hline
 &  \multicolumn{2}{c}{Design 1} & \multicolumn{2}{c}{Design 2}\\
\cline{2-3} \cline{4-5} 
Case Study &  AL & CP$\%$  & AL & CP$\%$ \\\hline\\

$\bar{Y}_{g1}\approx\bar{Y}_{g2}$ and $S_{g1}^2 \approx S_{g2}^2$ &  4.7309 & 88.81 & 4.7168 & 88.48  \\ \\
$\bar{Y}_{g1}\neq\bar{Y}_{g2}$ and $S_{g1}^2 \approx S_{g2}^2$ &  5.4929 & 92.35 & 5.0995 & 87.82 \\ \\
$\bar{Y}_{g1}\neq\bar{Y}_{g2}$  and $S_{g1}^2 \neq S_{g2}^2$ &  6.4273 & 90.90 & 6.0775 & 88.88 \\ \\
$\bar{Y}_{g1} \approx\bar{Y}_{g2}$ and $S_{g1}^2 \neq S_{g2}^2$ &  5.7853 & 89.32 & 5.7591 & 88.17  \\ \\
\hline
\end{tabular}
\end{table}

\section{E.B. and C.E.B. estimators for the Variance of Design 1}
\label{EBandCEB}

Let $s_g^2=(y_{g1}-y_{g2})^2/2$ denote the collapsed strata variance for the $g$-th group of strata. Here, $y_{gi}$ denotes the sampled observation from the $i$-th stratum in the $g$-th group. We assume $s_g^{2}/S_g^2$ has a chi-squared distribution with 1 degree of freedom $(\chi^2(1))$ and an inverse gamma prior Inv-Gamma(a,a) for $S_g^2$.
The posterior distribution $\pi(S^2_g|
s^2_g)$ is
\begin{equation*}
\pi(S^2_g|s^2_g) \propto f_{S^2_g}(s^2_g)\pi(S^2_g) \propto
\frac{(s^2_g)^{-1/2}e^{-s^2_g/2S^2_g}}{(S^2_g)^{1/2}}(S^2_g)^{-a-1}e^{-a/S^2_g}. 
\end{equation*}
This is an inverse gamma distribution with shape $a+\frac{1}{2}$ and scale $a+\frac{s^2_g}{2}$, i.e. Inv-Gamma$(a+\frac{1}{2},a+\frac{s^2_g}{2}).$ 

Under the squared error loss function, $L(S^2_g,\hat{\delta}^{B}_g)\equiv(S^2_g-
\hat{\delta}^{B}_g)^2$, the optimal Bayes estimator of $S^2_g$ is the expectation of $S^2_g$ conditional on $s^2_g$, i.e. $[E(S^2_g|s^2_g)]$, which is,
\begin{equation*}
\hat{\delta}^{B}_g=\lambda s^2_g+(1-\lambda)\frac{a}{a-1}=\frac{2a+s^2_g}{2a-1},
\end{equation*}
where $\lambda$ equals to $(2a-1)^{-1}$. 

We estimate parameter $``a"$ based on the method-of-moments, and therefore the marginal distribution of $s^2_g$ is needed. The marginal distribution of $s^2_g$ is the $F$-distribution with $1$ and $2a$ degrees of freedom or equivalently, $\sqrt{s^2_g}$ follows the Student's $t$-distribution with $2a$ degrees of freedom. The theoretical second order moment based on the Student's $t$-distribution is $E((\sqrt{s^2_g})^2)\equiv a/(a-1)$, which is valid for $a$ greater than $1$, and should be replaced by the empirical mean of the collapsed strata variances, $s^2_.=\sum_{g=1}^{H}s^2_g/H$. Therefore, the solution is $\hat{a}_{MM}=s^2_./(s^2_.-1),$
which yields an empirical Bayes estimator:
\begin{equation} \label{eq:4.1}
\hat{\delta}^{EB}_g=\frac{2\hat{a}_{MM}+s^2_g}{2\hat{a}_{MM}-1}.
\end{equation}
By substituting $\hat{\delta}^{EB}_g$ from expression \ref{eq:4.1} into \ref{eq:3.2}, the optimal estimator for the variance of one PSU per stratum design is attained:
\begin{equation} \label{eq:4.2}
\tilde{v}(\bar{y}_{st})=\frac{1}{2H^2}
\sum_{g=1}^{H}\hat{\delta}^{EB}_g. 
\end{equation}

In Design1 since the sample size from each stratum is small, the direct estimator $s^2_g$ is over-dispersed; therefore, under our proposed Bayesian model, we can show that:
\begin{equation*}
E\Big\{\frac{1}{H-1}\sum_{g=1}^{H}(s_g^2-s^2_.)^2\Big\} =\frac{a^2(2a-1)}{(a-1)^2(a-2)}>\frac{a^2}{(a-1)^2(a-2)}=
E\Big\{\frac{1}{H-1}\sum_{g=1}^{H}(S^2_g-S^2_.)^2\Big\},
\end{equation*}
where $S^2_.=\sum_{g=1}^{H}S^2_g/H$.
While the direct estimator $s^2_g$ shows over-dispersion, the Bayes estimator shows under-dispersion explained as follows.
\begin{lemma}
The MSE of Bayes estimator is smaller than the direct estimator, i.e.  \newline $E\{\sum_{g=1}^{H}(S^2_g-S^2_.)^2\}>
E\{\sum_{g=1}^{H}(\hat{\delta}_g^B-\hat{\delta}^B_.)^2\}$. 
\end{lemma}
\begin{proof}
\begin{align*}
E\Big\{\frac{1}{H-1}\sum_{g=1}^{H}(S^2_g-S^2_.)^2|s^2_g\Big\} & =
\frac{1}{H-1}\sum_{g=1}^{H}E\Big\{(S^2_g-S^2_.)^2|s^2_g\Big\} \\
& =\frac{1}{H-1}\sum_{g=1}^{H}\Big\{V[S^2_g-S^2_.|s^2_g]
+(E[S^2_g-S^2_.|s^2_g])^2\Big\}\\
& =\frac{1}{H-1}\sum_{g=1}^{H}\Big\{V[S^2_g-S^2_.|s^2_g]
+(\hat{\delta}^B_g-\hat{\delta}^B_.)^2\Big\} \\
& =\frac{1}{H-1}\sum_{g=1}^{H}V[S^2_g-S^2_.|s^2_g]+\frac{1}{H-1}\sum_{g=1}^{H}(\hat{\delta}^B_g-\hat{\delta}^B_.)^2 \\
& >\frac{1}{H-1}\sum_{g=1}^{H}(\hat{\delta}_g^B-\hat{\delta}^B_.)^2,
\end{align*}
where $\hat{\delta}^B_.=\sum_{g=1}^{H}\hat{\delta}^B_g/H$. Hence, $E\{\sum_{g=1}^{H}(S^2_g-S^2_.)^2\}>
E\{\sum_{g=1}^{H}(\hat{\delta}_g^B-\hat{\delta}^B_.)^2\}$.
\end{proof}
These inequalities hold when we use empirical Bayes estimator as well (see [\citep*{lahiri1990adjusted}] as an example). 
The problem of under-dispersion for these set of Bayes estimators is related to the fact that the standard Bayes estimator and empirical Bayes estimator shrink towards the prior mean $a/(a-1)$ specifically when the sample size is small.
A solution to this problem might be to attach more weight to the direct estimator. To do so, we can match the ensemble variances by minimizing the posterior expected squared error loss $E\{\sum_{g=1}^{H}(S^2_g-\sigma^2_g)^2|s^2_g\}$ subject to the following constraints:

\begin{enumerate}
\item[i)] 
$\sigma^2_.=\frac{1}{H}\sum_{g=1}^{H}\sigma^2_g=
\frac{1}{H}\sum_{g=1}^{H}\hat{\delta}^B_g=\hat{\delta}^B_.,$

\item[ii)] 
$\frac{1}{H-1}\sum_{g=1}^{H}(\sigma^2_g-\sigma^2_.)^2
=E\{\frac{1}{H-1}\sum_{g=1}^{H}(S^2_g-S^2_.)^2|s^2_g\}.$
\end{enumerate}

We therefore can write the posterior expected squared error loss $E\{\sum_{g=1}^{H}(S^2_g-\sigma^2_g)^2|s^2_g\}$ as:

\begin{equation} \label{eq:4.3}
E\Big\{\sum_{g=1}^{H}(S^2_g-\sigma^2_g)^2|s^2_g\Big\}=E\Big\{\sum_{g=1}^{H}(S^2_g-\hat{\delta}^B_g)^2\Big\}+
\sum_{g=1}^{H}(\hat{\delta}^B_g-\sigma^2_g)^2.
\end{equation}
In order to minimize the posterior expected squared error loss, it is sufficient to minimize the last term of \ref{eq:4.3}, which is the only term related to $\sigma^2_g$.
Using Lagrange multipliers $\lambda_1$ and $\lambda_2$, we minimize $\sum_{g=1}^{H}(\hat{\delta}^B_g-\sigma^2_g)^2$ subject to the constraints $\sum_{g=1}^{H}\sigma^2_g=C_1$ and $\sum_{g=1}^{H}(\sigma^2_g-\sigma^2_.)^2=C_2$ or $\sum_{g=1}^{H}\sigma^4_g=C_2+C^2_1/H$, which means we minimize the objective function
\begin{equation*}
\Phi=\sum_{g=1}^{H}(\hat{\delta}_g^B-\sigma^2_g)^2-\lambda_1(\sum_{g=1}^{H}
\sigma^2_g-C_1)-\lambda_2(\sum_{g=1}^{H}\sigma^4_g-C_2
-\frac{C_1^2}{H}),
\end{equation*}
with respect to the $\sigma^2_g$'s. Therefore, we can get
\begin{equation} \label{eq:4.4}
\sigma^2_{g,opt}=\frac{1}{1-\lambda_2}(\hat{\delta}^B_g+\frac{\lambda_1}{2}).
\end{equation}

By imposing the constraints on \ref{eq:4.4}, we obtain 
\begin{equation*}
\lambda_1=2\Big\{(1-\lambda_2)\frac{C_1}{H}-\hat{\delta}^B_.\Big\}, \quad\quad\
\lambda_2=1-\Big\{\sum_{g=1}^{H}(\hat{\delta}^B_g-\hat{\delta}^B_.)^2/C_2\Big\}^{1/2}.
\end{equation*}
We can rewrite $C_1$ and $C_2$ as follows:
\begin{equation*}
C_1=H\hat{\delta}^B_., 
\quad\
C_2=E\Big\{\sum_{g=1}^{H}(S^2_g-S^2_.)^2|s^2_g\Big\}.
\end{equation*}
Now by substituting $\lambda_1$ and $\lambda_2$ into \ref{eq:4.4}, we get the constrained Bayes estimator:
\begin{equation*}
\hat{\delta}^{CB}_{g}=\sigma^2_{g,opt}=\hat{\delta}^{B}_.+
\Bigg\{\frac{E[\sum_{g=1}^{H}(S^2_g-S^2_.)^2|s^2_g]}{\sum_{g=1}^{H}(\hat{\delta}^B_g-\hat{\delta}^B_.)^2}\Bigg\}^{1/2}
(\hat{\delta}^B_g-\hat{\delta}^B_.),
\end{equation*}
and using some algebra, we have
\begin{equation} \label{eq:4.5}
\hat{\delta}^{CB}_{g}=\hat{\delta}^{B}_{.}+\Bigg\{1+
\frac{\frac{1}{H}\sum_{g=1}^{H}V(S^2_g|s^2_g)}
{\frac{1}{H-1}\sum_{g=1}^{H}
(\hat{\delta}^B_g-\hat{\delta}^B_.)^2}\Bigg\}^{1/2}(\hat{\delta}^B_g-\hat{\delta}^B_.).
\end{equation}

In expression \ref{eq:4.5}, $V(S^2_g|s^2_g)$ is the posterior variance for the variance of Inv-Gamma$(a+\frac{1}{2},a+\frac{s^2_g}{2})$, which is $2(2a+s^2_g)^2/(2a-1)^2(2a-3)$.
Therefore, after some algebra, $\hat{\delta}^{CB}_g$ in \ref{eq:4.5} can be written as
\begin{equation} \label{eq:4.6}
\hat{\delta}^{CB}_g=\frac{1}{2a-1}\Bigg\{2a+s^2_.+(s^2_g-s^2_.)
\Bigg[1+\frac{8(H-1)(a^2+s^2_.a+\frac{1}{4H}
\sum_{g=1}^{H}(s^2_g)^2)}{(2a-3)
\sum_{g=1}^{H}(s^2_g-s^2_.)^2}\Bigg]^{1/2}\Bigg\}.
\end{equation}
Since the constrained empirical Bayes estimator is close to the constrained Bayes estimator and the empirical Bayes estimator is close to the Bayes estimator, all of the mentioned results can be applied to the empirical Bayes and constrained empirical Bayes estimators as well (for more details on the equivalency between Bayes estimator (or constrained Bayes estimator) and empirical Bayes estimator (or constrained empirical Bayes estimator), readers can consult [\citep*{ghosh1987robust}] and [\citep*{lahiri1990adjusted}]).

As a consequence, the constrained empirical Bayes $\hat{\delta}^{CEB}_g$ can be obtained by substituting $\hat{a}_{MM}$ into \ref{eq:4.6}, which gives:
\begin{align*} \label{eq:4.7}
\hat{\delta}^{CEB}_g & =\frac{1}{2\hat{a}_{MM}-1}\Bigg\{2\hat{a}_{MM}+s^2_.+(s^2_g-s^2_.) \\
& \quad \times \Bigg[1+\frac{8(H-1)(\hat{a}_{MM}^2+s^2_.\hat{a}_{MM}+\frac{1}{4H}
\sum_{g=1}^{H}(s^2_g)^2)}{(2\hat{a}_{MM}-3)
\sum_{g=1}^{H}(s^2_g-s^2_.)^2}\Bigg]^{1/2}\Bigg\}. \numberthis
\end{align*}
As there is a posibility of receiving negative values for \ref{eq:4.7}, which could be related to $(s^2_g-s^2_.)$, we therefore use $\hat{\delta}_g^{EB}$ for the negative situations.
Finally, the optimal estimator for the variance of one PSU per stratum design based on the constrained empirical Bayes is:
\begin{equation} \label{eq:4.8}
\tilde{v}_2(\bar{y}_{st})=\frac{1}{2H^2}\sum_{g=1}^{H}
\hat{\delta}^{CEB}_{g}.
\end{equation}

\section{Simulation Study and Results} \label{allsimulations}

For the sake of suitable comparisons among our candidate estimators \ref{eq:3.2}, \ref{eq:4.2}, and \ref{eq:4.8}, we conduct a Monte Carlo simulation study with 10,000 replications based on our proposed population in Table \ref{table3}. The empirical relative mean squared error (RMSE) was found using the following formula:
\begin{equation*}
\sum_{r=1}^{10,000}\frac{RMSE_r}{10,000};
\quad\
RMSE_r=\frac{\sqrt{(v^*(\bar{y}_{st,r})-V(\bar{Y}_{st}))^2}}{V(\bar{Y}_{st})}, 
\end{equation*}
where $v^*(\bar{y}_{st,r})$ takes the values of our candidate estimators \ref{eq:3.2}, \ref{eq:4.2}, and \ref{eq:4.8}, and $V(\bar{Y}_{st})$ is the randomization-based variance in \ref{eq:2.1}.

As we have different constraints for parameter $a$, theoretically and applicably, to make our estimators \ref{eq:4.1} and \ref{eq:4.7} become valid, we apply the intersection of constraints. So, the truncated $\hat{a}^{*}_{MM}$=$max(x,\hat{a}_{MM})$ is used, where $x=1.5+e$, and $e$ takes values greater than zero. We consider different values for $e$ to appropriately study the behaviors and/or effects of $\hat{a}^{*}_{MM}$, which is related to the lower bound $(x)$, on estimators \ref{eq:4.1}, \ref{eq:4.7}, and RMSEs as well. The results of comparisons are shown in Figure \ref{figure1}. 

According to the plots in Figure \ref{figure1}, the RMSEs of constrained empirical Bayes for all of the situations are greater than their competitors when $x$ is really close to $1.5$, since the denominator of \ref{eq:4.7} tends to zero for $\hat{\delta}^{CEB}_g$'s with $x$ greater than $\hat{a}_{MM}$, and since the quantity of $\hat{\delta}^{CEB}_g$'s with $x>\hat{a}_{MM}$ is reasonable; this can tremendously affect the results of RMSEs. 

After moving away from this crucial threshold $(1.5)$, the RMSE of constrained empirical Bayes decreases compared to the RMSE of empirical Bayes. However, by increasing the value of $x$ and assigning more weight $(2a-2)/(2a-1)$ to the prior mean, the constrained empirical Bayes cannot perform well. Imposing great values for the constrain of $a$ result in the under-dispersion of empirical and constrained empirical Bayes estimators; therefore, their RMSEs will be increased. In addition, when the means of collapsed strata within the groups are different, empirical Bayes and constrained empirical Bayes estimators outperform the classical collapsed stratum variance estimator.

\begin{figure}[h!]
\centering
\begin{tabular}{c}
\includegraphics[width=0.78\textwidth]{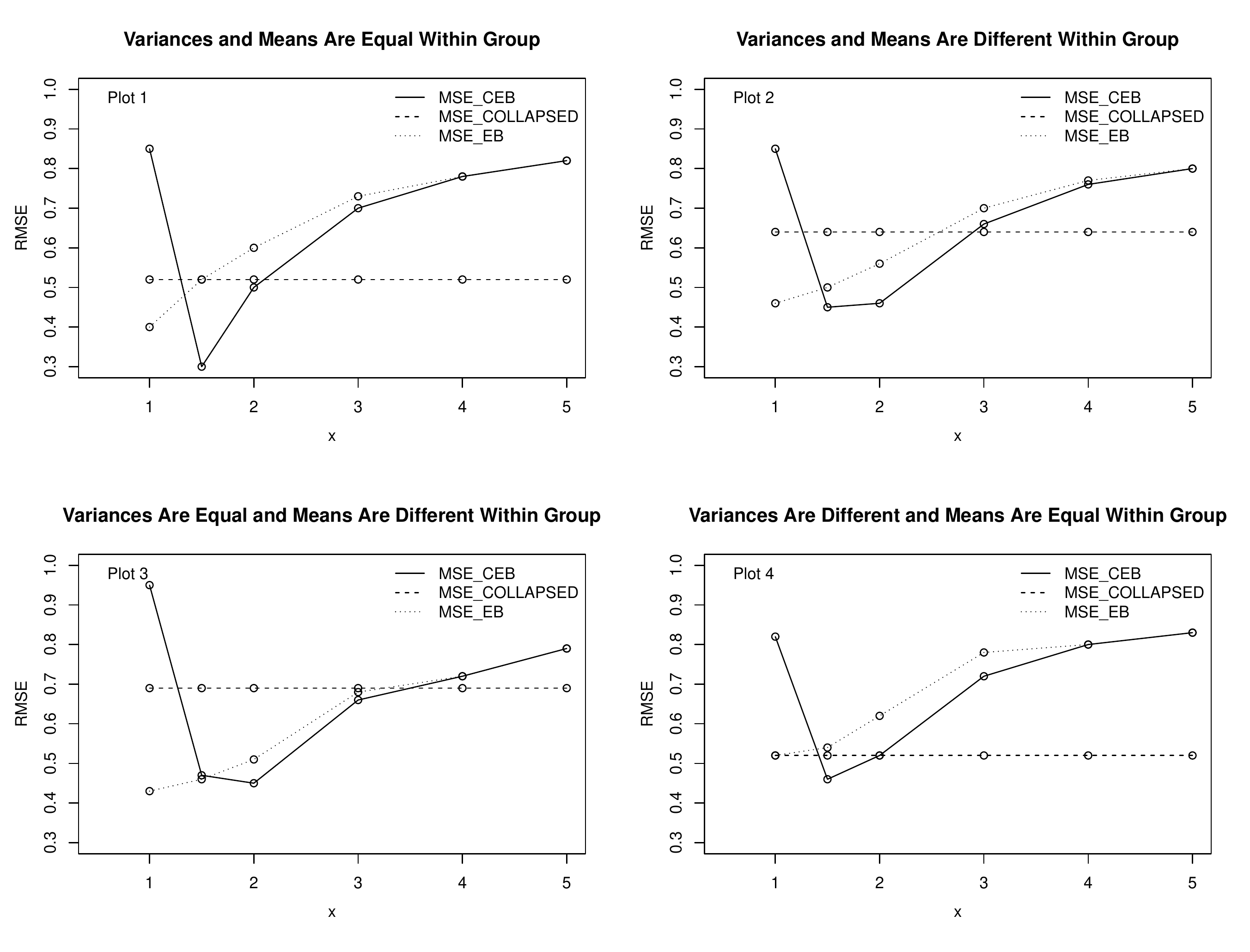} \\
\includegraphics[width=0.78\textwidth]{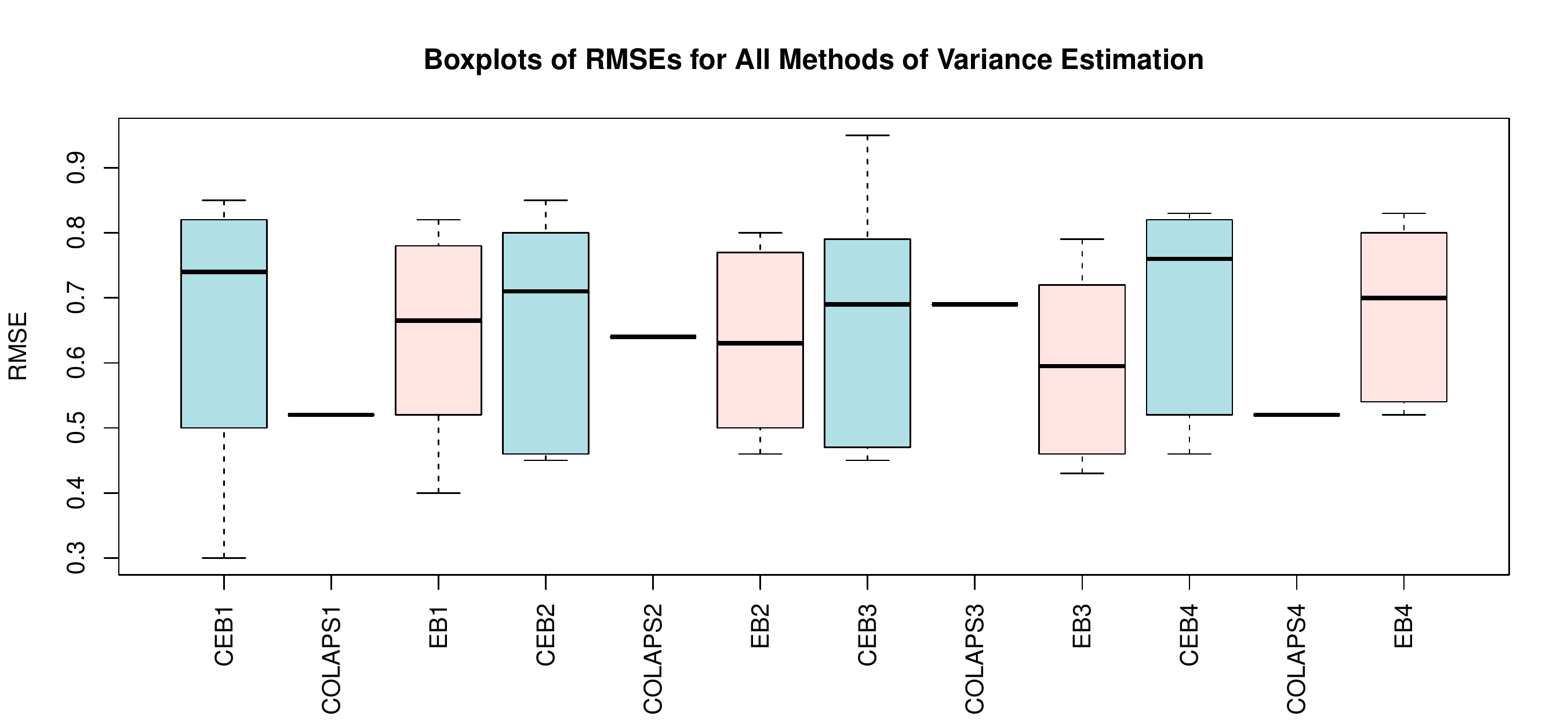}
\end{tabular}
\caption{Comparison Results of the Candidate Variances Based on the Empirical Relative MSE.}
\label{figure1}
\end{figure}

\section{Concluding Remarks} \label{remarks}

One PSU per stratum design has the advantage of deep stratification, which is efficient for estimating the finite population parameter of interest, but it is not possible to estimate the variance without making any implicit or explicit assumptions.
The collapsed stratum variance estimator, a classical method for estimating the variance of the design, usually suffers from the overestimation. 

In this paper, for obtaining the exact MSE expression (except for the FPC) for the collapsed variance estimator, we assum a single element has been selected using a simple random sampling within each stratum.  This assumption was made for the simplicity in exposition.  The MSE for collapsed variance estimator to the case when the single PSU is selected using a general design can be extended in a straightforward way.  The empirical Bayes and constrained empirical Bayes approaches developed in the paper are found to be promising alternatives to the traditional collapsed strata variance estimator for the one PSU per stratum design.  The estimation of the prior parameter $a$ in the EB and CEB approaches is a challenging problem.  In order to ensure that the estimator of the prior parameter $a$  is always within the admissible range, we proposed certain truncation strategies.  In the future, we plan to explore a hierarchical Bayesian approach in an effort to rectify the problem.

\newpage
\clearpage

\bibliographystyle{siam}
\bibliography{Bibliography}

\section{Appendix A: Proofs} \label{appendix}

All of the proofs in this section are based on the design-based application without consideration of any models.\\

\noindent \textbf{Expression \ref{eq:3.3}:}
\begin{proof}
\begin{align*}
Bias(v(\bar{y}_{st})) & =E(\frac{1}{2H^2}\sum_{g=1}^{H}s_g^2)-\frac{1}{4H^2}\sum_{h=1}^{2H}S_h^2
=\frac{1}{2H^2}\sum_{g=1}^{H}E(s^2_g)-\frac{1}{4H^2}\sum_{h=1}^{2H}S_h^2
\\
& =\frac{1}{2H^2}\sum_{g=1}^{H}E(\sum_{i=1}^{2}(y_{gi}-\bar{y}_{g})^2)-\frac{1}{4H^2}\sum_{h=1}^{2H}S_h^2 \\
 & =\frac{1}{2H^2}\sum_{g=1}^{H}E\{\frac{1}{2}(y_{g1}-y_{g2})^2\}-\frac{1}{4H^2}\sum_{h=1}^{2H}S_h^2 \\
& =\frac{1}{4H^2}\sum_{g=1}^{H}E\{(y_{g1}-\bar{Y}_{g1})-(y_{g2}-\bar{Y}_{g2})+(\bar{Y}_{g1}-\bar{Y}_{g2})\}^2 \\
& \quad-\frac{1}{4H^2}\sum_{h=1}^{2H}S_h^2=\frac{1}{4H^2}\sum_{g=1}^{H}E\{(y_{g1}-\bar{Y}_{g1})^2+(y_{g2}-\bar{Y}_{g2})^2 \\
& \quad +(\bar{Y}_{g1}-\bar{Y}_{g2})^2-2(y_{g1}-\bar{Y}_{g1})(y_{g2}-\bar{Y}_{g2})+2(y_{g1}-\bar{Y}_{g1}) \\
& \quad \times(\bar{Y}_{g1}-\bar{Y}_{g2})-2(y_{g2}-\bar{Y}_{g2})(\bar{Y}_{g1}-\bar{Y}_{g2})\}-\frac{1}{4H^2}\sum_{h=1}^{2H}S_h^2. 
\end{align*}

\noindent Under the stratified simple random sampling without replacement design, samples per stratum are selected independently; thus, $y_{g1}$ and $y_{g2}$ are independent. Also, $\bar{Y}_{g1}$ and $\bar{Y}_{g2}$, the population means in each collapsed stratum of group $g$ are fixed. As a result $E(y_{g1}-\bar{Y}_{g1})(y_{g2}-\bar{Y}_{g2})$ equals to $0$. Furthermore, $E(y_{g1})=\bar{Y}_{g1}$ and $E(y_{g2})=\bar{Y}_{g2}$, so we can rewrite $E(s_g^2)$ as follows:
\begin{equation*}
E(s_g^2) = \frac{1}{2}{E\{(y_{g1}-\bar{Y}_{g1})^2
+(y_{g2}-\bar{Y}_{g2})^2+(\bar{Y}_{g1}-\bar{Y}_{g2})^2\}}. 
\end{equation*}          
\noindent Thus 
\begin{align*}
Bias(v(\bar{y}_{st})) & =\frac{1}{4H^2}\sum_{g=1}^{H}
E\{(y_{g1}-\bar{Y}_{g1})^2+(y_{g2}-\bar{Y}_{g2})^2+
(\bar{Y}_{g1} - \bar{Y}_{g2})^2\} \\
& \quad-\frac{1}{4H^2}\sum_{h=1}^{2H}S_h^2=\frac{1}{4H^2}\sum_{g=1}^{H}
\{S^2_{g1}+S^2_{g2}+(\bar{Y}_{g1}-\bar{Y}_{g2})^2\}-\frac{1}{4H^2}\sum_{h=1}^{2H}S_h^2. 
\end{align*}                
Additionally, 
\begin{equation*}
E(y_{g1}-\bar{Y}_{g1})^2=\sum_{j=1}^{N}
(Y_{g1j}-\bar{Y}_{g1})^2/N=(1-N^{-1})S^2_{g1} \approx S^2_{g1}, 
\end{equation*}
\begin{equation*}
E(y_{g2}-\bar{Y}_{g2})^2=\sum_{j=1}^{N}
(Y_{g2j}-\bar{Y}_{g2})^2/N=(1-N^{-1})S^2_{g2}\approx S^2_{g2}. 
\end{equation*}

\noindent Therefore as $\sum_{g=1}^{H}(S^2_{g1}
+S^2_{g2})=\sum_{h=1}^{2H}S^2_h$, the bias is:
$$
Bias(v(\bar{y}_{st}))=\frac{1}{4H^2}\sum_{g=1}^{H}\{(\bar{Y}_{g1}-\bar{Y}_{g2})^2\}.
$$

\end{proof}

\noindent \textbf{Expression \ref{eq:3.4}:}
\begin{proof}
\begin{align*}
Var(v(\bar{y}_{st})) & =Var(\frac{1}{2H^2}\sum_{g=1}^{H}s^2_g)=\frac{1}{4H^4}\sum_{g=1}^{H}Var(s^2_g)=\frac{1}{4H^4}\sum_{g=1}^{H}Var(\sum_{i=1}^{2}(y_{gi}-\bar{y}_{g})^2) \\
 & =\frac{1}{4H^4}\sum_{g=1}^{H}Var(\frac{1}{2}(y_{g1}-y_{g2})^2) \\
 & =\frac{1}{16H^4}\sum_{g=1}^{H}\{E(y_{g1}-y_{g2})^4
-\{E(y_{g1}-y_{g2})^2\}^2\}, 
\end{align*}

\noindent where
\begin{align*}
E(y_{g1}-y_{g2})^4 & =E\{(y_{g1}-\bar{Y}_{g1})^4+(y_{g2}-\bar{Y}_{g2})^4
+(\bar{Y}_{g1}-\bar{Y}_{g2})^4 \\
& \quad +6(y_{g1}-\bar{Y}_{g1})^2(y_{g2}-\bar{Y}_{g2})^2
+6(y_{g1}-\bar{Y}_{g1})^2(\bar{Y}_{g1}-\bar{Y}_{g2})^2 \\
& \quad +6(y_{g2}-\bar{Y}_{g2})^2(\bar{Y}_{g1}-\bar{Y}_{g2})^2
+4(y_{g1}-\bar{Y}_{g1})^3(\bar{Y}_{g1}-\bar{Y}_{g2}) \\
& \quad-4(y_{g2}-\bar{Y}_{g2})^3(\bar{Y}_{g1}-\bar{Y}_{g2})\}=\mu_{4,g1}+\mu_{4,g2}+(\bar{Y}_{g1} - \bar{Y}_{g2})^4 \\
 & \quad+6S^2_{g1}S^2_{g2}+6S^2_{g1}(\bar{Y}_{g1}-\bar{Y}_{g2})^2+6S^2_{g2}(\bar{Y}_{g1}-\bar{Y}_{g2})^2 \\
& \quad+4\mu_{3,g1}(\bar{Y}_{g1}-\bar{Y}_{g2})-4\mu_{3,g2}(\bar{Y}_{g1}-\bar{Y}_{g2}), 
\end{align*}
and $\{E(y_{g1}-y_{g2})^2\}^2=\{S^2_{g1}+S^2_{g2}+(\bar{Y}_{g1}-\bar{Y}_{g2})^2\}^2.$ Therefore;
\begin{align*}
Var(v(\bar{y}_{st})) & =\frac{1}{16H^4}\sum_{g=1}^{H}
\{\mu_{4,g1}+\mu_{4,g2}+(\bar{Y}_{g1}-\bar{Y}_{g2})^4+
6S^2_{g1}S^2_{g2} \\
& \quad+6S^2_{g1}(\bar{Y}_{g1}-\bar{Y}_{g2})^2+6S^2_{g2}(\bar{Y}_{g1}-\bar{Y}_{g2})^2 \\
& \quad+4\mu_{3,g1}(\bar{Y}_{g1}-\bar{Y}_{g2})-4\mu_{3,g2}(\bar{Y}_{g1}-\bar{Y}_{g2}) \\
& \quad -\{S^2_{g1}+S^2_{g2}+(\bar{Y}_{g1}-\bar{Y}_{g2})^2\}^2\}. 
\end{align*}
As a result,
\begin{align*}
Var(v(\bar{y}_{st}))&=\frac{1}{16H^4}\sum_{g=1}^{H}\{\mu_{4,g1}
+\mu_{4,g2}+2S^2_{g1}S^2_{g2} \\
 & \quad+4(\bar{Y}_{g1}-\bar{Y}_{g2})^2(S^2_{g1}+S^2_{g2}) \\
& \quad-(S^2_{g1}-S^2_{g2})^2+4(\bar{Y}_{g1}-\bar{Y}_{g2})(\mu_{3,g1}-\mu_{3,g2})\}
\end{align*}
since
\begin{align*}
\mu_{3,g1}=E(y_{g1}-\bar{Y}_{g1})^3\quad\quad ,\quad\quad
\mu_{3,g2}=E(y_{g2}-\bar{Y}_{g2})^3 
\end{align*}
\begin{align*}
\mu_{4,g1}=E(y_{g1}-\bar{Y}_{g1})^4\quad\quad ,\quad\quad
\mu_{4,g2}=E(y_{g2}-\bar{Y}_{g2})^4. 
\end{align*}
\end{proof}

\vspace{0.5cm}

\noindent \textbf{Note:} This paper was published as a \textit{Proceeding in the Survey Research Methods Section, JSM 2015, American Statistical Association}.
\end{document}